\documentclass[10pt, conference, letterpaper]{IEEEtran}
\pdfoutput=1

\makeatletter
\long\def\@makecaption#1#2{\ifx\@captype\@IEEEtablestring%
\footnotesize\begin{center}{\normalfont\footnotesize #1}\\
{\normalfont\footnotesize\scshape #2}\end{center}%
\@IEEEtablecaptionsepspace
\else
\@IEEEfigurecaptionsepspace
\setbox\@tempboxa\hbox{\normalfont\footnotesize {#1.}~~ #2}%
\ifdim \wd\@tempboxa >\hsize%
\setbox\@tempboxa\hbox{\normalfont\footnotesize {#1.}~~ }%
\parbox[t]{\hsize}{\normalfont\footnotesize \noindent\unhbox\@tempboxa#2}%
\else
\hbox to\hsize{\normalfont\footnotesize\hfil\box\@tempboxa\hfil}\fi\fi}
\makeatother

\usepackage{bbm}
\usepackage{graphicx}
\usepackage{epstopdf}
\graphicspath{{figs/}}
\usepackage[font=footnotesize]{caption}
\usepackage{subfigure}
\usepackage{cite}
\usepackage{amsthm}
\usepackage{amssymb}
\usepackage{amsmath}
\usepackage{verbatim}
\usepackage{cite}
\usepackage{graphicx}
\usepackage{url}
\usepackage{setspace}

\usepackage[noend]{algpseudocode}
\usepackage{algorithm}

\newcommand{\etal}{\textit{et~al.}}
\newcommand{\eg}{\textit{e.g.},}
\newcommand{\ie}{\textit{i.e.},}

\newcommand{\dhit}{d^h_{im}}
\newcommand{\dmiss}{d^m_{im}}
\newcommand{\duc}{d^b_i}

\newtheorem{problem}{Problem}
\newtheorem{definition}{Definition}
\newtheorem{proposition}{Proposition}
\newtheorem{conjecture}{Conjecture}
\newtheorem{lemma}{Lemma}

\title{On the Complexity of Optimal Routing and Content Caching in Heterogeneous Networks} 

\author{Mostafa Dehghan$^1$, Anand Seetharam$^2$, Bo Jiang$^1$, Ting He$^3$, Theodoros Salonidis$^3$, \\ Jim Kurose$^1$, Don Towsley$^1$, and Ramesh Sitaraman$^{1,4}$\\

$^1$University of Massachusetts Amherst,
$^2$California State University Monterey Bay, \\
$^3$IBM T.J. Watson Research Center, 
$^4$Akamai Technologies Inc \\
{\tt \{mdehghan, bjiang, kurose, towsley, ramesh\}@cs.umass.edu}, \\ {\tt aseetharam@csumb.edu}, {\tt \{the, tsaloni\}@us.ibm.com}
}

\begin{document}
\maketitle


\vspace{-0.35cm}
\begin{abstract}
We investigate the problem of optimal request routing and content caching in a heterogeneous network supporting in-network content caching with the goal of minimizing average content access delay. 
Here, content can either be accessed directly from a back-end server (where content resides permanently) or be obtained from one of multiple in-network caches. 
To access a piece of content, a user must decide whether to route its request to a cache or to the back-end server. Additionally, caches must decide which content to cache. 
We investigate the problem complexity of two problem formulations, where the direct path to the back-end server is modeled as  {\it i)} a congestion-sensitive or {\it ii)} a congestion-insensitive path, reflecting whether or not the delay of the uncached path to the back-end server depends on the user request load, respectively.
We show that the problem is NP-complete in both cases. We prove that under the congestion-insensitive model the problem can be solved optimally in polynomial time if each piece of content is requested by only one user, or when there are at most two caches in the network.
We also identify a structural property of the user-cache graph that potentially makes the problem NP-complete.
For the congestion-sensitive model, we prove that the problem remains NP-complete even if there is only one cache in the network and each content is requested by only one user.  
We show that approximate solutions can be found for both models within a $(1-1/e)$ factor of the optimal solution, and demonstrate a greedy algorithm that is found to be within 1\% of optimal for small problem sizes.
Through trace-driven simulations we evaluate the performance of our greedy algorithms, which show up to a $50\%$ reduction in average delay over solutions based on LRU content caching.
\end{abstract}

\section{Introduction} \label{sec:introduction}

In-network content caching has received considerable attention in recent years as a means to address the explosive growth in data access seen in today's networks. Its main premise is to store content at the network's edge -- close to the end users -- to reduce  user content access delay and network bandwidth usage.  The benefits of in-network content caching have been demonstrated in the context of CDN~\cite{Huang08, Sharma13, Sitaraman10} as well as hybrid networks comprised of cellular and MANETs or femto-cell networks~\cite{Westphal13, Golrezaei13, Poularakis13}.

%
%

In this paper, we investigate a {\it joint\/} problem of in-network content caching and request routing in a hybrid network where stored content can be accessed through multiple heterogeneous network paths. 
We consider a scenario in which users send requests for content that is always available at a remote back-end server located in the network core but may also be present at multiple in-network caches. Access to the back-end server employs a potentially costly, congested, and/or slower uncached path, while the in-network caches may be reached through cheaper and faster network paths. This scenario arises in various hybrid network contexts where content can be accessed through multiple heterogeneous paths, including core/edge CDNs, macro/femto cell networks, cellular/MANET networks (\eg\ where the path to the network core is over cellular infrastructure and in-network caches are accessible via MANET paths), and cloud/edge cellular networks with edge storage at the cellular base stations. If a request is routed to an in-network cache that has the requested content, the request is served immediately. Otherwise, the cache must download the content from the back-end server before serving it to the user, incurring additional delay. Additionally, the cache must decide whether or not to store the downloaded content.

We address the following question: how should users route their requests among the paths to in-network caches and the back-end server, and what in-network cache management policy should  be adopted to minimize the average content access delay across all users? 
We consider two variants of the problem. First, we consider a \emph{congestion-insensitive} delay model (termed \emph{CI-model}), assuming that delays are independent of the traffic load on all paths. Second, we consider a \emph{congestion-sensitive} delay model (termed \emph{CS-model}), assuming that the delay to the back-end server (\ie\ the uncached path) depends on the traffic load. In a hybrid cellular/MANET network, the uncached path in the CI-model corresponds to GBR (guaranteed bit rate) 3GPP bearer service, while in the CS-model it corresponds to Non-GBR Aggregate Maximum Bit Rate (AMBR) bearer service~\cite{ixia}.


Our goal in this paper is two-pronged. First, we seek a principled understanding of the computational complexity of the joint caching and routing problem: i) Can the general problem be solved optimally in polynomial time? ii) If not, are there problem instances that are tractable and which of the above modeling aspects make the general problem intractable? Second, we seek efficient approximate solutions to the joint routing/caching problem, with approximation guarantees, that work well in practice.

Toward our first goal, we provide a unified optimization formulation of the joint caching and routing problem for both models and show this problem is NP-complete in the general case. Then, we investigate which factors contribute to the problem complexity. 
For the CI-model, we prove that the optimal solution can be found in polynomial time in two special cases: a) when each user requests a single piece of content, or b) when there are at most two in-network caches.
We also identify a condition which is potentially the root cause of the complexity of the problem in the general case $-$ cycles with an odd number of users and caches in a graph that represents the network.
For the CS-model, we prove that the problem is ``harder'': it remains NP-complete even if there is a single cache and each user accesses a single distinct content.
These results provide valuable insights on the problem complexity. 

Toward our second goal, we show that the problem of optimal joint caching and routing for both the CI-model and the CS-model can be formulated as maximization of a monotone submodular function subject to matroid constraints. This enables us to devise two greedy algorithms. The first one has a higher complexity but can produce solutions within a $(1  - 1/e)$  factor of the optimal solution for both the CI-model and CS-model. The second algorithm has a lower complexity but does not have known approximation guarantees.  
We evaluate the performance of these algorithms through numerical evaluations and trace-driven simulations on a large dataset of approximately 9 million requests for 3 million content items. The results show that both algorithms are within $1\%$ of the optimal for small problem sizes where computing the optimal solution is feasible and that significant reductions (up to $50\%$) in content access delay can be achieved over traditional LRU-based content caching schemes.

%

Our contributions can be summarized as follows:
\begin{itemize}
\item We provide a unified optimization formulation for the joint caching and routing problem for the CI-model and the CS-model and prove that the problem is NP-complete in both cases.


\item We derive insights into problem complexity by considering several special cases, some of which are shown to admit efficient solutions, while others remain computationally hard.  

\item We develop a greedy caching and routing algorithm that achieves an average delay within a $(1 - 1/e)$ factor of the optimal solution and a second greedy algorithm of lower complexity. 

\item We evaluate the performance of these algorithms through numerical evaluations and trace-driven simulations.
Numerical results show that the greedy algorithms perform close to the optimal solution when computing the optimal solution is feasible.
Our results from trace-driven simulations show that the greedy algorithms yield significant performance improvement compared to solutions based on traditional LRU caching policy.
\end{itemize}

The paper is organized as follows. The network model and the joint caching and routing problem formulation are presented in Sections~\ref{sec:model} and ~\ref{sec:problem}. Sections~\ref{sec:constant} and~\ref{sec:mm1}, present the complexity results for the congestion-insensitive and congestion-sensitive cases, respectively. The approximation algorithms are presented in Section~\ref{sec:approx} and their performance is evaluated in Section~\ref{sec:trace}. Section~\ref{sec:relatedWork} reviews related work, and Section~\ref{sec:conclusion} concludes the paper.

\section{Network Model}
\label{sec:model}

\begin{figure}[t]
\begin{center}
  	\includegraphics[scale=0.25]{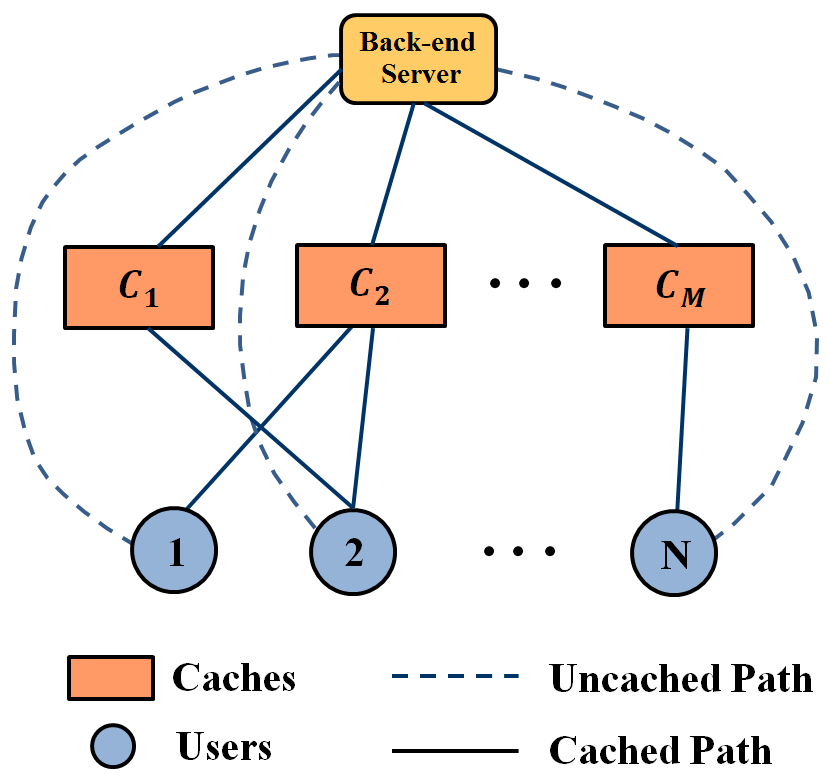} 
  \caption{Hybrid network with in-network caching}
 \vspace{-0.5cm}
    \label{fig:system}
\end{center}
\end{figure}

In this section, we consider the network shown in Figure~\ref{fig:system} with  $N$ users generating requests for a set of $K$ unique files $F~=~\{f_1, f_2, \ldots, f_K\}$ of unit size. Throughout this paper, we will use the terms content and file interchangeably. We assume that these files reside permanently at the back-end server. As shown in Figure~\ref{fig:system}, there are  $M$ caches in the network that can serve user requests.

All files are available at the back-end server and users are directly connected to this server via a cellular infrastructure. We refer to the cellular path between the user and the back-end server as the \emph{uncached path}. Each user can also access a subset of the $M$ in-network caches where the content might be cached. We refer to the connection between the user and a cache as a \emph{cached path}.

Let $C_m$ denote the storage capacity of the $m$-th cache measured by the maximum number of files it can store. If user $i$ requests file $j$ and it is present in the cache, then the request is served immediately. We refer to this event as a cache hit. However, if content $j$ is not present in the cache, the cache then forwards the request to the back-end server, downloads file $j$ from the back-end server and forwards it to the user. We refer to this event as a cache miss, since it was necessary to download content from the back-end server in order to satisfy the request. \emph{Note that in case of a cache miss, the cache can decide whether to keep the downloaded content.}

User $i$ generates requests for the files in $F$ according to a Poisson process of aggregate rate $\lambda_i$. Aggregate request rate of all users is $\lambda$. We denote by $q_{ij}$ the probability that user $i$ generates a request for file $j$ (referred to as the \emph{file popularity}). The popularity of the same file can vary from one user to another.

Let $A = [a_{im}]$  denote the connections between users and caches, with $a_{im} = 1$ if user $i$ is connected to cache $m$, and $a_{im} = 0$, otherwise. For the user-cache connections, let $\dhit$ and $\dmiss$ denote the average delays incurred by user $i$ in the event of a cache hit or miss at cache $m$, respectively. We assume without loss of generality that  $\dmiss  > \dhit$, \ie\ cache misses always incur greater delays than cache hits. We consider two models for the delay over the path from users to the back-end server. The first is a congestion-insensitive, constant-delay model where the delays through the uncached path are independent of the traffic load on the link to the back-end server. 
 In this case, the average delay experienced for a request by user $i$ sent over the uncached path is $\duc$.
The second model is a congestion-sensitive delay model where delays experienced over the uncached path depend on the traffic load. In this case, we assume the back-end server has service rate $\mu$, and model the connections to the back-end server as an M/M/1 queue. The delay experienced over the uncached path then consists of an initial access delay with average $\duc$ and a queuing (waiting plus service) delay with average $1/(\mu - \lambda_q)$, assuming $\lambda_q$ is the request rate on the queue.
 

\section{Problem Formulation}
\label{sec:problem}

In this work, we consider a joint caching and routing problem with the goal of minimizing average content access delay over the requests of all users for all files. The solution to this problem requires addressing two closely-related questions 1) How should cache contents be managed -– which files should be kept in the caches, and what cache replacement strategy should be used? and 2) How should users route their requests between the cached and uncached paths?

For our routing policy, we define the decision variable $p_{ijm}$ that denotes the fraction of the requests of user $i$ for content $j$ sent to cache $m$. User $i$ sends the remaining $1 - \sum_m{p_{ijm}}$ fraction of her requests for content $j$ to the back-end server through the uncached path.

It is shown in~\cite{Liu_98} for a single cache that given a routing policy, \emph{static caching} achieves minimum expected delay. With static caching, a set of files is stored in the cache, and the cache content does not change in the event of a cache hit or miss. The argument in~\cite{Liu_98} was extended in~\cite{milcom} to a network of caches to show that static caching achieves minimum delay under a fixed routing policy. Hence, we define the binary variables $x_{jm} \in \{0, 1\}$ to denote the content placement in caches, where $x_{jm}=1$ indicates file $j$ is stored in cache $m$ and $x_{jm}=0$ indicates otherwise.


We denote by $D(\mathbf{x}, \mathbf{p})$ the expected delay obtained by content placement strategy $\mathbf{x} = [x_{jm}]$, and routing strategy $\mathbf{p} = [p_{ijm}]$. The optimal solution to the problem of joint caching and routing is therefore obtained by solving the following Mixed-Integer Program (MIP):
\begin{equation}
\begin{aligned}
\label{eq:main_opt}
\text{minimize } & D(\mathbf{x}, \mathbf{p}) \\
\text{such that } & \sum_m{p_{ijm}} \le 1 \quad &&\forall{i,j}\\
 & \sum_j{x_{jm}} \le C_m &&\forall{m}\\
 & x_{jm} \in \{0, 1\} &&\forall{j, m}\\
 & 0 \le p_{ijm} \le a_{im} &&\forall{i,j,m}
\end{aligned}
\end{equation}

In the next two sections, we express the delay function $D(\mathbf{x}, \mathbf{p})$ for the cases of {\it i)} congestion-insensitive and {\it ii)} congestion-sensitive uncached path delay models, and discuss why the joint caching and routing problem is NP-complete.




\section{Congestion-Insensitive Uncached Path}
\label{sec:constant}
First, we consider the case where delays on the uncached path, $\duc$, do not depend on the traffic load on the back-end server. For a given content placement $\mathbf{x}$ and routing policy $\mathbf{p}$, the average delay can be written as
\begin{align}
\label{eq:opt_constant}
D(\mathbf{x}, \mathbf{p}) &= \frac{1}{\lambda}\sum_i\sum_j\lambda_i q_{ij} \Biggl( \sum_m{p_{ijm}x_{jm}\dhit} \notag\\
+& \sum_m{p_{ijm}(1-x_{jm})\dmiss} + \Big( 1 - \sum_m{p_{ijm}}\Big)\duc \Biggr)
\end{align}

Without loss of generality, we assume that $\dhit~<~\duc~<~\dmiss$ whenever user $i$ is connected to cache $m$, \ie\ $a_{im} = 1$. Note that if $\duc \ge \dmiss, \forall{i}$, users connected to cache $m$ will never use the uncached path. Also, if $\duc \le \dhit, \forall{i}$, none of the users will actually use cache $m$.

It is easy to see that with the congestion-insensitive model, given a content placement, the average minimum delay is obtained by routing requests for the cached content to caches, and routing the remaining requests to the uncached path. Note that under this routing policy no cache misses occur. Therefore, the solution to the problem of joint caching and routing in the case of congestion-insensitive uncached path delays are obtained by solving the following binary linear program:
\begin{equation}
\label{eq:opt_simpl_constant}
\begin{split}
\text{minimize } & \frac{1}{\lambda}\sum_i{\lambda_i \sum_j{q_{ij} \left[ \sum_m{p_{ijm} \dhit}  + (1 - \sum_m{p_{ijm}}) \duc \right]}} \\ 
\text{such that } & \sum_m{p_{ijm}} \le 1 \\
 & \sum_j{x_{jm}} \le C_m \\
 & p_{ijm} \le x_{jm} \cdot a_{im} \\
 & x_{jm} \in \{0, 1\} \\
 & p_{ijm} \ge 0.
\end{split}
\end{equation}
Note that $D(\mathbf{x}, \mathbf{p})$ is a linear function of the routing variables. Also note the additional constraint $p_{ijm} \le x_{jm} \cdot a_{im}$, which is due to the fact that only requests for cached content are routed to caches. Since $\dhit < \duc$ and $\duc < \dmiss$, users have no incentive to split the traffic for any content between the cached and uncached paths, and hence there will be no routing variable, $p_{ijm}$, with a fractional value in the optimal solution, \ie\ $p_{ijm}\in \{0,1\}$.

\subsection{Hardness of General Case}
The above formulation of the joint caching and routing problem is a generalization of the Helper Decision Problem (HDP) proved to be NP-complete in~\cite{Golrezaei12}. Our formulation is more general as we consider non-homogeneous delays for the cached and uncached paths. HDP reduces to the optimization problem in~\eqref{eq:opt_simpl_constant} by setting $\duc = 1$, $\dhit = 0$, and $C_m = C$, where $C$ is the cache size at all caches in HDP.



Although the problem is NP-complete in general, we will show that the joint caching and routing problem can be solved in polynomial time for several special cases, and discuss what makes the problem ``hard'' in general.
We first consider a restrictive setting where each user is interested in only one file and each file is requested by only one user. Next, we consider a network with two caches (but each user may be interested in an arbitrary number of files). We present polynomial time solution algorithms for both cases. Finally, we present an example that demonstrates what we conjecture to be {\it the} source of the complexity of this problem.

\subsection{Special Case: One File per User}
Consider the network illustrated in Figure~\ref{fig:system}, but assume each user is interested in only one file, \ie\ $q_{ii} = 1$, and $q_{ij} = 0$ for $i\neq j$. In this case, the optimal solution to the joint caching and routing problem can be found in polynomial time based on a solution to the \emph{maximum weighted matching} problem.

Note that in this case, the number of files equals the number of users, \ie\ $N=K$. To avoid triviality, we assume that the number of users is larger than the capacity of each cache in the network, \ie\ $C_m < N, \forall{m}$. The assumption that each user is interested in only one file allows us to re-write the objective function in~\eqref{eq:opt_simpl_constant} as
\begin{equation*}
\begin{split}
D(\mathbf{x}, \mathbf{p}) &= \frac{1}{\lambda}\sum_i{\Big( \sum_m{\lambda_i p_{iim} \dhit} + \lambda_i(1-\sum_m{p_{iim}})\duc \Big)} \\
&= \frac{1}{\lambda} \Big( \sum_{i = 1}^{N}{\lambda_i \duc} - \sum_i{\sum_m{\lambda_i p_{iim} (\duc - \dhit)}} \Big).
\end{split}
\end{equation*}
Since $\sum_{i = 1}^{N}{\lambda_i \duc}$ is a constant independent of the decision variables, minimizing the above objective function is equivalent to maximizing $\sum_i{\sum_m{\lambda_i p_{iim} (\duc - \dhit)}}$. Note that $\lambda_i(\duc - \dhit)$ can be interpreted as the gain obtained by having file $i$ in cache $m$. This problem can then be naturally seen as matching files to caches with the goal of maximizing the sum of individual gains. In what follows, we map this problem to the maximum weighted matching problem.

For each cache of size $C_m$, we introduce $C_m$ nodes $\{v_m^1, v_m^2, \ldots, v_m^{C_m}\}$ representing unit size micro-caches that form cache $m$. Let $V = \{v_1^1, v_1^2, \ldots, v_1^{C_1},\ldots, v_M^1, \ldots, v_M^{C_M}\}$ denote the set of of all such nodes, and let $U~=~\{u_1,u_2,\ldots,u_N\}$ denote the set of all files.
We define the bipartite graph $G(U, V, E)$ with $\lambda_i(\duc~-~\dhit)$ as the weight of the edges connecting node $u_i$ to nodes $v^s_m, \forall{s\in\{1, 2, \dots, C_m\}}$. Figure~\ref{fig:matching} demonstrates a bipartite graph with user/file nodes $u$ and the micro-cache nodes $v$ with the edge weights shown for some of the edges.
Note that the bipartite graph consists of $|U| + |V| = N + \sum_m{C_m}$ vertices and $|E|~=~O(N\sum_m{C_m})$ edges.

The optimal solution to the joint content placement and routing problem corresponds to the maximum weighted matching for graph $G$. The edges selected in the maximum matching determine what content should be placed in which cache. Users then route to caches for cached content, and to the uncached path for the remaining files.

The maximum weighted matching problem for bipartite graphs can be solved in $O(|V|^2|E|)$ using the Hungarian algorithm~\cite{west01}. In our context, the complexity is $O(M^3N^4)$. Note that $\sum_m{C_m} = O(MN)$ as we assumed $C_m < N, \forall{m}$. Therefore, we can solve the joint caching and routing problem in polynomial time when users are interested in one file only.

\begin{figure} [ht]
  \begin{center}
  \includegraphics[scale=0.25]{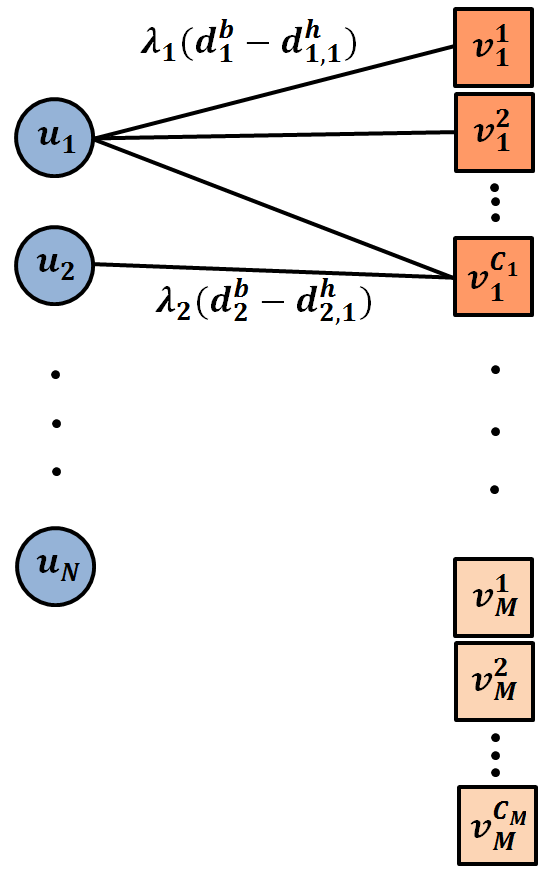}
  \caption{Modeling content placement as a maximum weighted matching problem}
  \vspace{-0.5cm}
    \label{fig:matching}
  \end{center}
\end{figure}

\subsection{Special Case: Network with Two Caches}

Next, we show that the optimal solution for the joint caching and routing
problem can be found in polynomial time when there are only two caches in the network.
Specifically, we prove that the solution to the integer program~\eqref{eq:opt_simpl_constant} can be found in polynomial time when there are two caches in the network. In the remainder of this section, we assume that $x_{jm}$ take real values.

Before delving into the proof we introduce some definitions and results
from~\cite{hoffman10}:

\begin{definition}
A square integer matrix is called \emph{unimodular} if it has determinant $+1$ or $-1$.
\end{definition}

\begin{definition}
An $m \times n$ integral matrix $A$ is \emph{totally unimodular} if
the determinant of every square submatrix is $0$, $1$, or $-1$.
\end{definition}

\begin{proposition}
\label{prop:integral}
If for a linear program $\{\max c^Tx : Ax \le b\}$, $A$ is totally unimodular and $b$ is integral, then there is an optimal solution to the linear program that is integral.
\end{proposition}



Note that the three sets of constraints in the optimization problem
in~\eqref{eq:opt_simpl_constant}, namely, i) $\sum_m{p_{ijm}} \le 1$, ii) $\sum_j{x_{jm}} \le C_m$, and iii) $p_{ijm} - x_{jm} \cdot a_{im} \le 0$ can be written in the form $A z \le b$ where the
entries of $A$ and $b$ are all integers, and $z$ consists of the $x_{jm}$ and $p_{ijm}$ entries. 
From Proposition~\ref{prop:integral}, then, it suffices to show that
the matrix $A$ is totally unimodular for a network with two caches
to prove that optimization~\eqref{eq:opt_simpl_constant} can be
solved in polynomial time. To prove that the matrix $A$ is totally
unimodular we use the following result from~\cite{schrijver03}:
\begin{proposition}
\label{prop:TU}
A matrix is totally unimodular if and only if for every subset $R$ of rows, there is an assignment $s:R \to \pm1$ of signs to rows so that the signed sum $\sum_{r \in R} s(r)r$ (which is a row vector of the same width as the matrix) has all its entries in $\{0,\pm 1\}$.
\end{proposition}

In Appendix~\ref{sec:appndx_two_cache}, we give a constructive proof showing that for any subset $R$ of rows
of $A$ we can find an assignment $s$ that satisfies Proposition~\ref{prop:TU}.

%
%


\subsection{Complexity Discussion}

\begin{figure}[t]
\begin{center}
\includegraphics[scale=0.3]{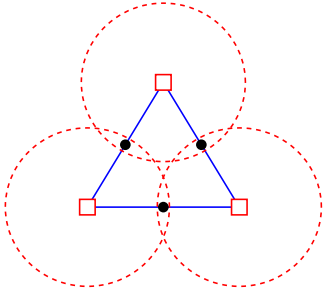}
\caption{A network with three users and three caches. Each user is in the communication range of two of the caches.}
\label{fig:cycle_network}
\end{center}
\end{figure}

\begin{figure}
\centering
\begin{subfigure}{(a)} 
  \centering
  \includegraphics[scale=0.15]{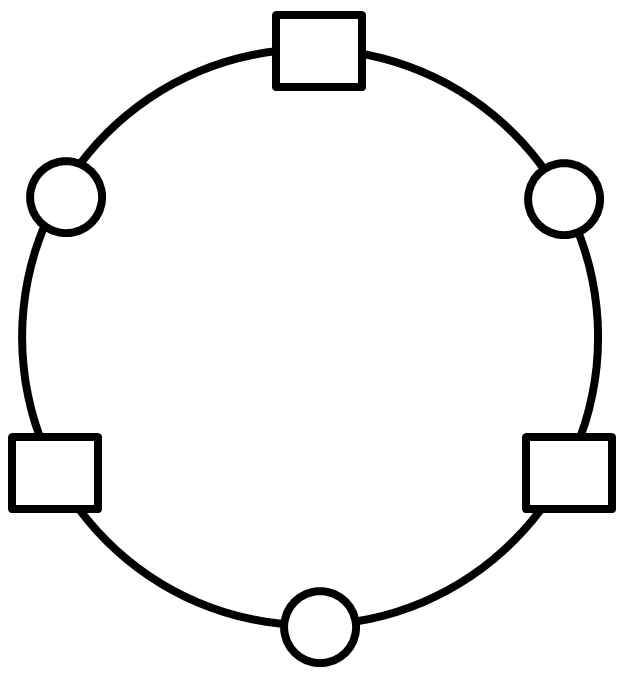}
  \label{fig:cycle_sub1}
\end{subfigure}%
\hspace{5mm}
\begin{subfigure}{(b)}
  \centering
  \includegraphics[scale=0.15]{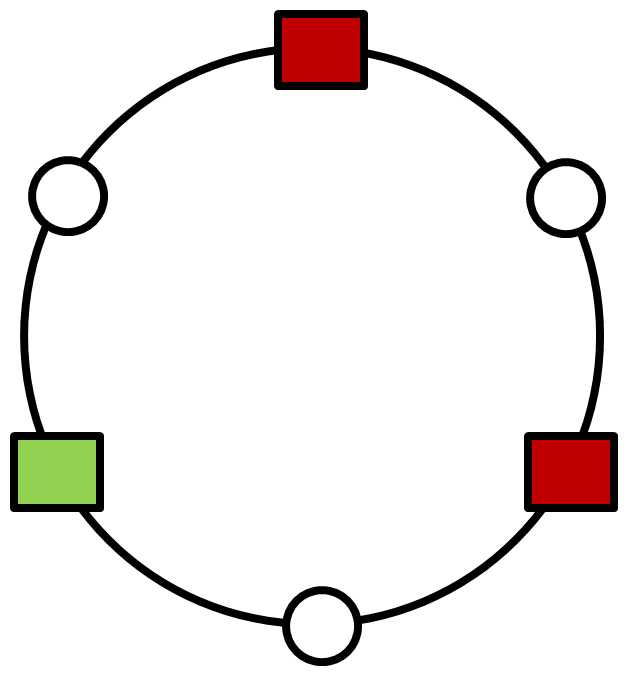}
  \label{fig:cycle_sub2}
\end{subfigure}

\begin{subfigure}{(c)}
  \centering
  \includegraphics[scale=0.15]{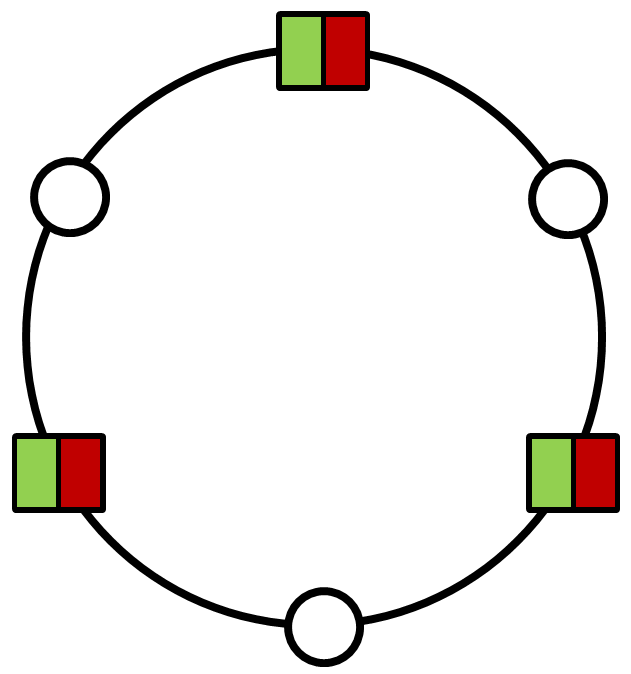}
  \label{fig:cycle_sub3}
\end{subfigure}%
\hspace{5mm}
\begin{subfigure}{(d)}
  \centering
  \includegraphics[scale=0.15]{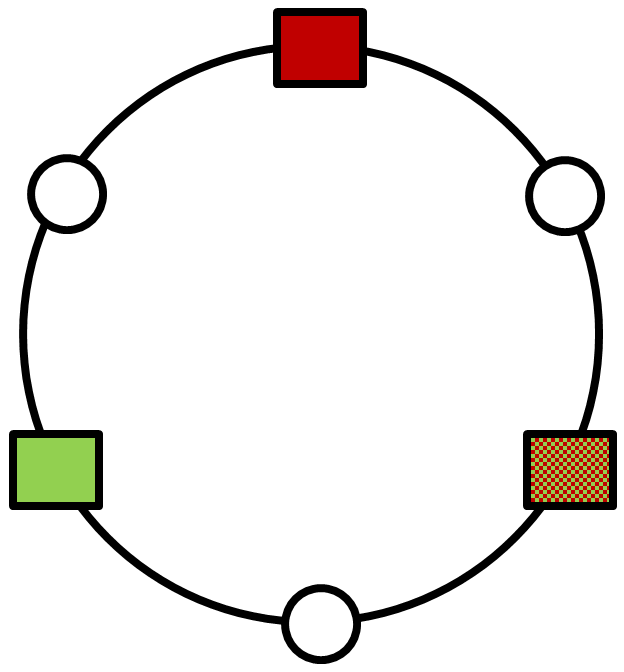}
  \label{fig:cycle_sub4}
\end{subfigure}
\caption{(a) A network of three users connected to three caches forming a cycle. (b) Optimal content placement according to binary placement decisions, \ie\ $x_{jm}\in \{0, 1\}$. (c) Optimal content placement assuming fractions of files can be stored in caches, \ie\ $0\le x_{jm}\le 1$. (d) Optimal content placement with the possibility of content coding.}
\label{fig:cycle}
\end{figure}

Although the problem of joint caching and routing is NP-complete in general, in the previous subsections we showed that several non-trivial special cases of the problem can be solved in polynomial time. In this section, we discuss the potential cause of complexity of the problem.

By relaxing the integer constraints on content placement variables, $x_{jm}$, and allowing them to take real values, \ie\ $0 \le x_{jm} \le 1$, we obtain another problem that is generally referred to as the ``relaxation'' of problem~\eqref{eq:opt_simpl_constant}. Since the objective function in~\eqref{eq:opt_simpl_constant} is convex, the solution to the \emph{relaxed problem} can be found in polynomial time for all instances of the problem.

By comparing the solutions to the integer and the relaxed problems for a large number of instances of the optimization problem in~\eqref{eq:opt_simpl_constant} we discuss what we believe is the root cause of the complexity of the joint caching and routing problem in the case of congestion-insensitive uncached path.
We observe via numerical evaluations that for most instances of the problem, solutions to problem~\eqref{eq:opt_simpl_constant} match those of the relaxed problem. Those instances that result in different solutions to the two problems exhibit a certain structure that we explain here.

Consider a network with three users and three caches as depicted in Figure~\ref{fig:cycle_network}. With each user connected to two of the caches, the user-cache connections can be seen to form a cycle as demonstrated in  Figure~\ref{fig:cycle}a. Assume all paths from users to caches have equal hit and miss delays. Also, assume that each cache has the capacity of storing one file, and that all three users are interested in two files, noted here as green and red. 

Solving the optimization problem~\eqref{eq:opt_simpl_constant} for the above network, the optimal content placement is to replicate one of the files in two of the caches, and have one copy of the other file in the third cache, as shown in Figure~\ref{fig:cycle}b. With this content placement, two of the users get both files from their cached paths, and the third user gets only one file from cached paths, and has to use the uncached path to get the other file.
The solution to the \emph{relaxed} optimization problem however would be to store half of each file in each cache, \ie\ $x_{1m}~=~x_{2m}~=~0.5$, which achieves strictly smaller average delay. This solution is illustrated in Figure~\ref{fig:cycle}c\footnote{Note that we do not consider the solution of the relaxed problem as a legitimate content placement. Although it looks like all users can access the two files via the caches in Figure~\ref{fig:cycle}c, when splitting the files in halves, two of the caches will store the same half copy of a file,   and the user connected to those caches will only get half of that file from the caches and still needs to use the uncached path for the other half. However, we acknowledge that with the possibility of coding, content placement can be done in such a way that users can get both files from caches, as is shown in Figure~\ref{fig:cycle}d. We are not considering coded content placement in this work.}.

The above discussion shows how the solution to optimization problem~\eqref{eq:opt_simpl_constant} differs from its relaxed counterpart for the network shown in Figure~\ref{fig:cycle_network}. It can easily be seen that matrix $A$ corresponding to the linear constraints $A x \le b$ for this network is \emph{not} totally unimodular.

Such mismatch between the solutions of the optimization problem~\eqref{eq:opt_simpl_constant} and the corresponding relaxed problem, is also observed for larger networks that contain odd number of users and odd number of caches connected in a way that form a cycle. It is easy to show for all these networks that the matrix $A$ of the constraints is \emph{not} totally unimodular. 

We conjecture that these cycles are \emph{the} source of complexity in the problem of joint caching and routing, and for networks that do not have any such cycles the solution to the optimization problem~\eqref{eq:opt_simpl_constant}  matches that of the relaxed problem. More specifically we have the following conjecture:

\begin{conjecture}
The optimal solution to the problem of joint caching and routing can be found
in polynomial time if there are no cycles of length $4k+2, k \ge 1$
in the bipartite graph corresponding to the user-cache connections. 
\end{conjecture}
\section{Congestion-Sensitive Uncached Path}
\label{sec:mm1}
Next, we consider the case where delays on the uncached path depend on the traffic load on the back-end server. We model the uncached path as an M/M/1 queue with service rate $\mu$. In addition to the queuing delay, we assume that user $i$ observes an initial access delay to uncached path with average $\duc, i=1,\ldots, N$. Here, we make no assumptions regarding $\duc$ and $\dhit$. Note that if $\dhit \le \duc$ and the needed object is in the cache user $i$ will direct all her requests for that object to cache. If $\dhit > \duc$, however, even if the needed content is in cache $m$, the user may prefer to use the uncached path, depending on the service rate and the load on the back-end server. For a given content placement $\mathbf{x}$ and routing policy $\mathbf{p}$, the average delay can be written as
\begin{align}
\label{eq:mm1_delay}
D(\mathbf{x}, \mathbf{p}) &= \frac{1}{\lambda}\Biggl[\sum_i\sum_j\lambda_i q_{ij} \Biggl( \sum_m{p_{ijm}x_{jm}\dhit} \notag\\
&+ \sum_m{(1-x_{jm})p_{ijm}\dmiss} + (1-\sum_m{p_{ijm}}) \duc \Biggr) \notag\\
&+ \frac{\sum_i\sum_j\lambda_i q_{ij} (1-\sum_m{p_{ijm}})}{\mu - \sum_i\sum_j\lambda_i q_{ij} (1-\sum_m{p_{ijm}})} \Biggr].
\end{align}

\subsection{Hardness of General Case}
Note that we can consider the congestion-insensitive delay model as a special case of the congestion-sensitive model where $\mu = +\infty$. This explains why this problem is NP-complete in general. In the remainder of this section, however, we will prove that the problem of joint caching and routing in the case of a congestion-sensitive delay model remains NP-complete even if there is only one cache in the network and each content is of interest to no more than one user.

\subsection{Hardness of Single-Cache Case}
Modifying the delay function $D(x, p)$ in~\eqref{eq:mm1_delay} for the case of one cache, \ie\ $M=1$, and assuming each user is interested in only one file, \ie\ $q_{ii} = 1, \forall{i}$, we can re-write the optimization problem as
\begin{equation}
\label{eq:single_cache}
\begin{aligned}
\text{minimize } \quad & \frac{1}{\lambda} \left[\sum_{i=1}^{N}{\lambda_i x_i p_i d_i^h} + \sum_{i=1}^{N}{ \lambda_i (1-x_i)p_i d_i^m} \right. \\
& \left. \quad + \sum_{i=1}^{N}{\lambda_i (1-p_i) \duc} + \frac{\sum_{i=1}^{N}{ \lambda_i (1-p_i) }}{\mu - \sum_{i=1}^{N}{ \lambda_i (1-p_i) } } \right] \\
\text{such that } \quad &\sum_{i=1}^{N}{x_i} \le C \\
&0 \le p_i \le a_i \\
&x_i \in \{0, 1\},
\end{aligned}
\end{equation}
where $p_i = p_{ii1}$ denotes the probability that user $i$ will use the uncached path. Also, $a_i$ denotes whether user $i$ is connected to the cache.

To show that the above optimization problem is NP-complete, we consider the corresponding
decision problem, Congestion Sensitive Delay Decision Problem (CSDDP).

\begin{problem}
(Congestion Sensitive Delay Decision Problem) Let $\mathbf{\Lambda} = [\lambda_1, \lambda_2, \ldots, \lambda_N]$ denote the request rates of users, and let $\mathbf{d}^h = [d_i^h]$, $\mathbf{d}^m = [d_i^m]$ and $\mathbf{d}^b = [d_i^b]$ denote the hit delay, miss delay and initial delay of uncached path, respectively. Also, let $\mu$ be the service rate of the back-end server, and $C$ be the cache capacity.\\
We are asking the following question: given the parameters $(\mu, \mathbf{\Lambda}, \mathbf{d}^h, \mathbf{d}^m, \mathbf{d}^b, C)$ and a real number $d$, is there any assignment of $\mathbf{x} = [x_i]$ and $\mathbf{p} = [p_i]$ such that $D(\mathbf{x}, \mathbf{p}) \le d$.
\end{problem}

It is clear that for any given content placement $\mathbf{x}$ and routing policy $\mathbf{p}$ 
the answer to CSDDP can be verified in polynomial time, and hence CSDDP is in class NP.
To prove that CSDDP is NP-hard, we use the fact that the following problem is NP-hard.

\begin{problem}
(Equal Cardinality Partition) Given a set $A$ of $n$ numbers, can $A$ be partitioned
into two disjoint subsets $A_1$ and $A_2$ such that $A=A_1\cup A_2$, the sum of the numbers in $A_1$
equals the sum of the numbers in $A_2$ and that $|A_1|=|A_2|$?
\end{problem}

\begin{lemma}
ECP is NP-hard.
\end{lemma}
\begin{proof}
A proof of NP-hardness of a more general form of ECP is given in~\cite{Cieliebak08}.
Here, we give a simpler proof by a reduction from the Partition problem.

\begin{problem}
(Partition) Given a set $A$ of $n$ positive integers, can $A$ be partitioned
into two disjoint subsets $A_1$ and $A_2$ such that $A = A_1\cup A_2$ and the sum of the numbers
in $A_1$ equals the sum of the numbers in $A_2$?
\end{problem}

For each instance of Partition with input $A = \{a_1, \ldots, a_n\}$ create an instance
$A' = \{a_1, \ldots, a_n, 0, \ldots, 0\}$ by adding $n$ zeros to $A$.
It is easy to see that $A'$ can be partitioned into two subsets with equal cardinality
if and only if $A$ can be partitioned. Therefore, Partition $\le_P$ ECP, and ECP is NP-hard.
\end{proof}

\begin{lemma}
CSDDP is NP-Complete.
\end{lemma}
\begin{proof}
See Appendix~\ref{sec:mm1_np} for a detailed proof.
\end{proof}

Although this problem is NP-complete even in a very restrictive case with one cache and each user requesting one file, in the next section we show that a greedy algorithm can find approximate solutions with guaranteed performance.
\section{Approximation Algorithms} 
\label{sec:approx}

In this section, we show that the problem of joint caching and routing (for both congestion-insensitive and congestion-sensitive delay models)
can be formulated as the maximization of a monotone submodular function subject to
matroid constraints. This enables us to devise algorithms
with provable approximation guarantees.

We first review the definition and properties of matroids~\cite{nemhauser88}, and monotone~\cite{bartle76} and submodular~\cite{schrijver03} functions, and then prove our problem can be formulated as the maximization of a monotone submodular function subject to
matroid constraints.

\begin{definition}
A matroid $M$ is a pair $M = (S, I)$, where $S$ is a finite set and $I \subseteq 2^S$
is a family of subsets of $S$ with the following properties:
\begin{enumerate}
\item $\O \in I$,
\item $I$ is downward closed, \ie\ if $Y \in I$ and $X \subseteq Y$, then $X\in I$,
\item If $X, Y\in I$, and $|X| < |Y|$, then $\exists y\in Y\backslash X$ such that
$X\cup\{y\}\in I$.
\end{enumerate}
\end{definition}

\begin{definition}
Let $S$ be a finite set. A set function $f:2^S\rightarrow\mathbb{R}$ is submodular
if for every $X, Y \subseteq S$ with $X\subseteq Y$ and every $x\in S\backslash Y$ we have
\[f(X\cup\{x\}) - f(X) \ge f(Y\cup\{x\}) - f(Y).\]
\end{definition}

\begin{definition}
A set function $f$ is monotone increasing if $X\subseteq Y$ implies that $f(X) \le f(Y)$.
\end{definition}

Let $X_m$ denote the set of files stored in cache $m$, and define $X~=~X_1\cup X_2\cup\ldots\cup X_M$ to be the set of files stored in the $M$ caches. $X$ is the set equivalent of the binary content placement $\mathbf{x}$ defined in~\eqref{eq:main_opt}. Note that $|X_m| \le C_m$ 


Let $S_m = \{s_{1m}, s_{2m}, \ldots, s_{Km}\}$ denote the set of all possible files that could be placed in cache $m$ where $s_{jm}$ denotes the storage of file $j$ in cache $m$.
The set element $s_{jm}$ corresponds to the binary variable $x_{jm}$ defined in the optimization problem~\eqref{eq:main_opt} such that $x_{jm} = 1$ if and only if the element $s_{jm}\in X$.
Define the super set $S~=~S_1\cup S_2\cup\ldots\cup S_M$ as the set of all possible content placements in the $M$ caches. We have the following lemma.

\begin{lemma}
The constraints in~\eqref{eq:main_opt} form a matroid on~$S$.
\end{lemma}

\begin{proof}
For a given content placement $\mathbf{x}$, the optimal routing
policy can be computed in polynomial time since
$D(\mathbf{p})~=~D(\mathbf{p}; \mathbf{x})$ is a convex function. With that in
mind, we can write the average delay as a function of the
content placement $X\subseteq S$. Thus, the constraints
on the capacities of the caches can be expressed as $X\subseteq \mathcal{I}$ where
\[\mathcal{I} = \{X\subseteq S: |X\cap S_m| \le C_m, \forall m=1,\ldots, M\}.\]
Note that $(S, \mathcal{I})$ defines a matroid.
\end{proof}

Let $d_{ij}(\mathbf{x})$ denote the minimum average delay for user $i$ accessing file $j$ through a cached path, given content placement $\mathbf{x}$. We have
\[d_{ij}(\mathbf{x}) = \min_{m}{d_{ijm}},\]
where $d_{ijm}$ denotes the average delay of accessing content $j$ from cache $m$ defined as ($x_{jm}$ indicates that file $j$ is in cache $m$)
\[d_{ijm} = \dhit x_{jm} + \dmiss (1-x_{jm}).\]

Given the content placement in the caches, let $p_{ij}$ denote the fraction of the traffic for which user $i$ uses the cached paths to access content $j$. We can re-write the delay functions~\eqref{eq:opt_constant} and~\eqref{eq:mm1_delay} for the congestion-insensitive and the congestion-sensitive models as
\begin{equation*}
\begin{aligned}
D(\mathbf{p}; \mathbf{x}) &= \frac{1}{\lambda}\left(\sum_{i, j}{\lambda_i q_{ij} p_{ij} d_{ij}(\mathbf{x})} + \sum_{i, j}{\lambda_i q_{ij} (1-p_{ij}) \duc} \right),
\end{aligned}
\end{equation*}
and
\begin{equation*}
\begin{aligned}
D(\mathbf{p}; \mathbf{x}) &= \frac{1}{\lambda}\left[\sum_{i, j}{\lambda_i q_{ij} p_{ij} d_{ij}(\mathbf{x})} + \sum_{i, j}{\lambda_i q_{ij} (1-p_{ij}) \duc} \right.\\
&\qquad \left.+ \frac{\mu}{\mu - \sum_{i,j}{\lambda_i q_{ij} (1-p_{ij})}} - 1 \right],
\end{aligned}
\end{equation*}
respectively. The optimal routing policy for a given content placement~$\mathbf{x}$, then, is one that maximizes $-D(\mathbf{p}; \mathbf{x})$, and can be found by solving the following optimization problem:
\begin{equation}
\label{eq:optimization}
\begin{aligned}
\text{maximize}& \quad -D(\mathbf{p}; \mathbf{x}) \\
\text{such that}& \quad 0 \le p_{ij} \le 1 \quad \forall{i,j} \\
\end{aligned}
\end{equation}

Let $\mathbf{x}_X$ be the equivalent binary representation of the content placement set $X$. We have the following lemma for both congestion-insensitive and congestion-sensitive delay models:
\begin{lemma}
\label{lem:F}
Let $\mathcal{P}$ denote all routing policies. For $X\subseteq S$, the function $F(X) = \max_{\mathbf{p}\in \mathcal{P}}{ -D(\mathbf{p}; \mathbf{x}_X) }$ is a monotone increasing and submodular function.
\end{lemma}
\begin{proof}
See the Appendix for a detailed proof.
\end{proof}

A direct consequence of Lemma~\ref{lem:F} is that minimizing the objective function in~\eqref{eq:opt_constant} or~\eqref{eq:mm1_delay} is equivalent to maximizing a monotone submodular function. Therefore, the approximate solution obtained by the greedy algorithm in Algorithm~\ref{alg:greedy} is within a $(1-1/e)$ factor of the optimal solution  (see~\cite{Calinescu07}).

Algorithm~\ref{alg:greedy} starts with empty caches and at each step
greedily adds a file to the cache that maximizes function $F$.
This process continues until all caches are filled to capacity.
Optimal routing is then determined based on the content
placement.

\begin{algorithm}
\caption{GreedyWG: A greedy approximation with performance guarantees.}
\label{alg:greedy}
\begin{algorithmic}[1]
	\Statex
		\State $S\gets \{s_{jm} : 1 \le j \le K, 1\le m\le M\}$
		\State $X_m \gets \O, \forall{m}$
		\State $X \gets \O$
		\For{$c \gets 1 \textbf{ to } \sum_m{C_m}$}
			\State $s_{j^* m^*} \gets \arg\max_{s_{jm}\in S}{F(X\cup \{s_{jm}\})}$
			\State $X_{m^*}\gets X_{m^*}\cup \{s_{j^*m^*}\}$
			\State $X\gets X\cup \{s_{j^*m^*}\}$
			\If{$|X_{m^*}| = C_{m^*}$}
				\State $S \gets S \backslash s_{jm^*}, \forall{j}$
			\Else
				\State $S \gets S \backslash s_{j^*m^*}$
			\EndIf
		\EndFor
		\State Content placement is done according to $X$.
		\State Determine the routing as $\mathbf{p}^* \gets \arg\min_{\mathbf{p}}{D(\mathbf{p}; \mathbf{x}_X)}$.
\end{algorithmic}
\end{algorithm}

Although the greedy algorithm in Algorithm~\ref{alg:greedy} is guaranteed to find solutions within a $(1-1/e)$ factor of the optimal solution, its complexity is high, $O(M^2N^2K^2\log{(NK)})$. We devise a second, computationally more efficient, greedy algorithm in Algorithm~\ref{alg:greedy2} with time complexity $O(M^3NK)$. We do not have accuracy guarantees for Algorithm~\ref{alg:greedy2}, but in the next section, we will show that it performs very well in practice.

Algorithm~\ref{alg:greedy2} is based on the following ideas. It starts with empty caches and initializes the cache access delays for users as the miss delays to their closest caches. Then at each step a file is greedily selected to be placed in a cache that maximizes the change in the user access delays, $\sum_i{\lambda_i q_{ij} (d_{ij} - \min\{d_{ij}, d^h_{im}\})}$. This process continues until the caches are filled. Finally, similar to Algorithm~\ref{alg:greedy}, a routing policy that minimizes $D(\mathbf{p}; \mathbf{x})$ is determined.

\begin{algorithm}
\caption{Greedy: A greedy approximation without known performance guarantees.}
\label{alg:greedy2}
\begin{algorithmic}[1]
	\Statex
		\State $X_m \gets \O, \forall{m}$
		\State $X \gets \O$
		\State $d_{ij} \gets \min_c\{d^m_{ic}\}, \forall{i, j}$
		\For{$c \gets 1 \textbf{ to } \sum_m{C_m}$}
			\State $G_{jm} \gets [0]_{K\times M}$
			\For{$m \gets 1 \textbf{ to } M$}
				\If{$|X_m| < C_m$}
					\For{$j \gets 1 \textbf{ to } K$}
						\State $G_{jm} \gets \sum_i{\lambda_i q_{ij} (d_{ij} - \min\{d_{ij}, d^h_{im}\})}$
					\EndFor
				\EndIf			
			\EndFor
			\State $[j^*, m^*] \gets \arg\max_{j,m}{G_{jm}}$
			\State $X_{m^*}\gets X_{m^*}\cup \{s_{j^*m^*}\}$
			\State $X \gets X \cup \{s_{j^*m^*}\}$
			\State $d_{ij^*} \gets \min\{d_{ij^*}, d^h_{im^*}\}, \forall{i}$
		\EndFor
		\State Content placement is done according to $X$.
		\State Determine the routing as $\mathbf{p}^* \gets \arg\min_{\mathbf{p}}{D(\mathbf{p}; \mathbf{x}_X)}$
\end{algorithmic}
\end{algorithm}

\section{Performance Evaluation}
\label{sec:trace}

In this section, we evaluate the performance of the approximate algorithms. Our goal here is to evaluate 1) how well the solutions of greedy algorithms compare to the optimal (when computing the optimal solution is feasible), and 2) how well solutions from the greedy algorithms compare to those produced by a baseline.
Due to lack of space, we only consider the more realistic case of congestion-sensitive delay model.
For our baseline, we compare the approximate algorithms to the delay obtained by the following algorithm we will refer to as \emph{p-LRU}.

\subsection{p-LRU}
The cache replacement policy at all caches is Least Recently Used (LRU). For the routing policy, we assume that users that are not connected to any caches forward all their requests to the back-end servers. The remaining users, for each request, use a cached path with probability $p$ and with probability $1-p$ forward the request to the uncached path. If user $i$ decides to use a cached path, she chooses uniformly at random one of the $n_i$ caches she is connected to. The value of $p$ is the same for all users that have access to a  cache, and is optimized to minimize the average delay.

First, assuming users equally split their traffic across the caches that they can access, the aggregate popularity for individual files is computed at each cache. Let $r^m_j$ denote the normalized aggregate popularity of file $j$ at cache $m$. We have
\[r^m_j = \frac{1}{\Lambda}\sum_{i\in \mathcal{I}_m}{\lambda_i q_{ij} / n_i},\]
where $\mathcal{I}_m$ denotes the set of users connected to cache $m$, and $\Lambda$ is the normalizing constant across all files. Note that $r^m_j$ is independent of the parameter $p$.
With the aggregate popularities at hand, hit probabilities are computed at each cache using the \emph{characteristic time} approximation~\cite{Che01}.
Let $\mathbb{P}(x_{jm}=1)$ denote the probability that file $j$ resides in cache $m$. From~\cite{Che01} we have 
\[\mathbb{P}(x_{jm}=1) = 1 - \exp{(-r^m_j T_m)},\]
where $T_m$ is the characteristic time of cache $m$ is the unique solution to the equation
\[C_m = \sum_{j}{1 - \exp{(-r^m_j T_m)}}.\]
Given the cache hit probabilities, the average delay in accessing content $j$ from caches for user $i$ equals
\[d^c_{ij} = \frac{1}{n_i}\sum_{m\in\mathcal{M}_i}{\left[ \mathbb{P}(x_{jm}=1) \dhit + (1-\mathbb{P}(x_{jm}=1))\dmiss \right] },\]
where $\mathcal{M}_i$ denotes the set of caches that user $i$ is connected to. Note that $|\mathcal{M}_i| = n_i$.

Let $\mathcal{I}$ denote the set of users that are at least connected to one cache, and let $\lambda_{\mathcal{I}}$ denote the aggregate request rate of these users.
The average delay of accessing content from caches then equals
\[D_c = \frac{1}{\lambda_{\mathcal{I}}}\sum_{i\in\mathcal{I}}{\sum_j{\lambda_i q_{ij} d^c_{ij}}}.\]

Remember that some users might not be connected to any caches. Considering the traffic from all users, we can write the overall average delay as
\begin{align*}
D_{\text{LRU}} &= \frac{1}{\lambda} \left[ p \lambda_{\mathcal{I}} D_c + (1-p) \sum_{i\in\mathcal{I}}{\lambda_i \duc} + \sum_{i\not\in\mathcal{I}}{\lambda_i\duc} \right. \\
&\left. \qquad + \frac{\mu}{\mu - (1-p)\sum_{i\in\mathcal{I}}{\lambda_i} - \sum_{i\not\in\mathcal{I}}{\lambda_i}} - 1 \right].
\end{align*}
By differentiating $D_{\text{LRU}}$ with respect to $p$, the optimal value of $p$ is found to be
\[p^* = \max\{0, \min\{1,\left(\sqrt{\frac{\mu \sum_{i\in\mathcal{I}}{\lambda_i}}{\lambda_{\mathcal{I}} D_c - \sum_{i\in\mathcal{I}}{\lambda_i \duc}}} - \mu + \lambda \right) / \lambda_{\mathcal{I}}\}\}.\]

\subsection{Network Setup}
We consider a network with users uniformly distributed in a square field.
We consider two architectures. First, we assume there is only one large cache at the center of the network as in Figure~\ref{fig:trace_network}a. Second, we consider a network with five small caches with equal storage capacity as in Figure~\ref{fig:trace_network}b. 
Figure~\ref{fig:trace_network} also shows the communication range of the caches in each case.
In the single-cache network, the cache has a larger communication range and five times the capacity of each of the caches in the multi-cache network.

Users that are not in the communication range of any caches can only use the uncached path to the back-end server. The hit delay for each user is linearly proportional to the distance from the cache and has the maximum value\footnote{Here, delay aggregates all request propagation and download delays as well as the processing and queuing delays. We use normalized delay values instead of using any specific time unit.} of 12.5 time units and 5.5 time units for the single and multi-cache systems, respectively. For a cache miss, an additional delay of 25 time units is added to the hit delay.
The initial access delay of the uncached path is set to 5 time units for each user, and the service rate is proportional to the aggregate request rate, where the scaling factor will be specified later.

\begin{figure}[h]
\hspace{-0.275in}
\centering
  	\includegraphics[scale=0.4]{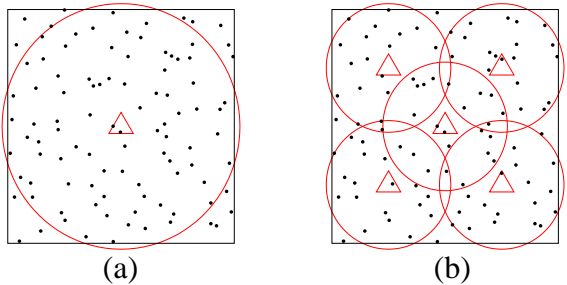}
\caption{A network with (a) one cache, and (b) five caches.}
    \label{fig:trace_network}
\end{figure}

\subsection{Numerical Evaluation}
\subsubsection{GreedyWG vs. Optimal}
First, we compare the solution of GreedyWG the approximate algorithm in Algorithm~\ref{alg:greedy} to the optimal solution. Due to the exponential complexity of finding the optimal solution, we are only able to compute the optimal solution for small problem instances. Here, we consider a network with five users and a single cache. User request rates are arbitrarily set to satisfy $\sum_i{\lambda_i} = 5$. We assume users are interested in 15 files, and that the aggregate user request popularities follow a Zipf distribution with skewness parameter~0.6. The service rate of the back-end server equals $\mu = 1$.


Figure~\ref{fig:single} shows the average delay and the $95\%$ confidence interval over 100 runs of each algorithm.
It is clear that GreedyWG performs very close to optimal. In fact, we observe that GreedyWG differs from the optimal solution in less than $20\%$ of the time, and the relative inaccuracy is never more than $1\%$.

\begin{figure}[h]
\begin{center}
  	\includegraphics[scale=0.5]{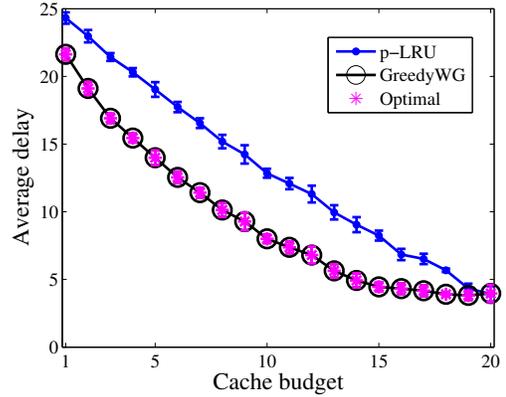} 
  \caption{Evaluation of GreedyWG against Optimal and p-LRU.}
    \label{fig:single}
\end{center}
\end{figure}

\subsubsection{GreedyWG vs. Greedy}
Next, we compare the solutions of GreedyWG against those of Greedy, the approximate algorithm, Algorithm~\ref{alg:greedy2}, with lower computational complexity but no performance guarantees.
We consider a network with five caches and 100 users uniformly distributed in a $10\times 10$ field.

Figure~\ref{fig:mutli_cache_size} shows the average delay and the $95\%$ confidence interval for different values of available cache budget.
Greedy (red curve) is barely distinguishable from GreedyWG (black curve), meaning that Greedy performs very close to GreedyWG.

We also evaluate these algorithms over different values of the service rate at the back-end server. Figure~\ref{fig:mutli_mu} shows the average delay
for $\mu$ between 2 to 7, with the aggregate traffic rate set to $\lambda = 5$.
Similar to Figure~\ref{fig:mutli_cache_size}, Greedy performs very close to GreedyWG, and is always within $1\%$ of GreedyWG.


\begin{figure}[h]
\centering
  	\includegraphics[scale=0.5]{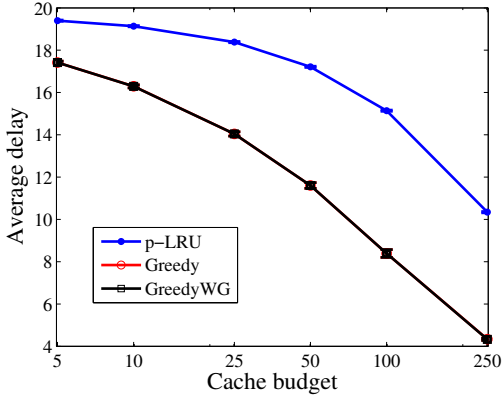} 
  \caption{Evaluation of the two greedy approximations over different values of the cache budget split equally between five caches. Aggregate user request rate is $\lambda= 5$,  and service rate of the back-end server equals 2.5.}
    \label{fig:mutli_cache_size}
\end{figure}

\begin{figure}[h]
\begin{center}
  	\includegraphics[scale=0.5]{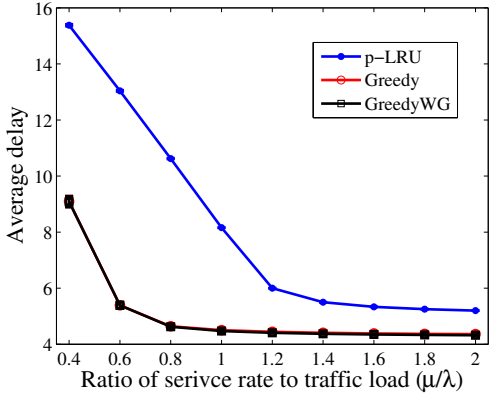} 
  \caption{Evaluation of the greedy algorithms for different values of the service rate at the back-end server. Aggregate user request rate is $\lambda= 5$, and the service rate varie from 2 to 10. Cache budget is set to 125.}
    \label{fig:mutli_mu}
\end{center}
\end{figure}


\subsection{Trace-driven Simulation}
Here, we present trace-driven evaluation results where we use traces for web accesses collected from an industrial research lab. 
The trace consists of approximately 9 million requests generated from 142,000 distinct IP addresses for more than 3 million distinct files. 
We only consider Greedy, the greedy algorithm presented in Algorithm~\ref{alg:greedy2}, since it performs close to Algorithm~\ref{alg:greedy}, and has lower complexity.


\begin{figure}[h]
\centering
  	\includegraphics[scale=0.5]{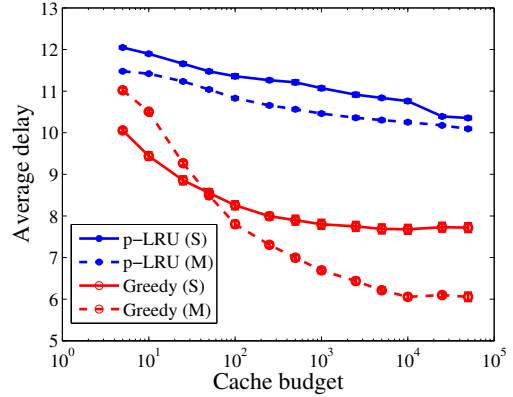} 
\caption{Evaluation of the Greedy and p-LRU for the single-cache (S) and multi-cache (M) network setups for different values of the available cache budget. The service rate is set to be 0.8 times the aggregate traffic rate.}
    \label{fig:trace_cache_size}
\end{figure}

\begin{figure}[h]
\centering
  	\includegraphics[scale=0.5]{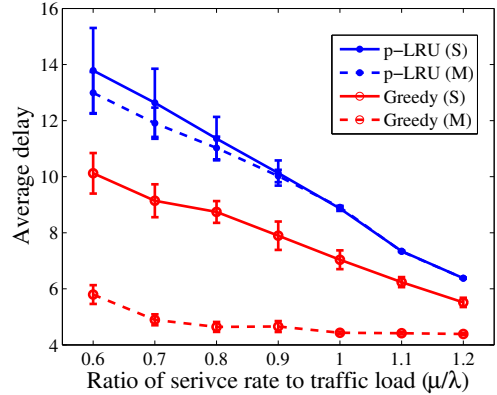} 
\caption{Evaluation of the Greedy and p-LRU algorithms for different values of the service rate to aggregate traffic ratio for the single-cache (S) and multi-cache (M) network setups.}
    \label{fig:trace_mu}
\end{figure}

To evaluate the Greedy algorithm using the trace data, we first divide the trace into smaller segments of approximately 120,000 requests. Each segment includes requests for approximately 40,000 distinct files, generated by approximately 2500 users. For every two consecutive segments, we use the first segment as learning dataset from which we compute the file popularities and determine the optimal value $p$ for the p-LRU scheme as well as content placement and routing based on the Greedy algorithm. We use the second segment to compute the average delays under the p-LRU and Greedy algorithms.

Figure~\ref{fig:trace_cache_size} compares the average delays for different cache budgets for the p-LRU and the Greedy algorithms for the single-cache (S) and multi-cache (M) networks. Significant reductions in average delay of up to $50\%$ are observed for both single-cache and multi-cache networks when using Greedy over p-LRU. While p-LRU yields similar performance in both single-cache and multi-cache architectures, the Greedy shows the advantage of one architecture over the other depending on the cache budget. When the cache budget is small, it is better to have a single cache with larger cache size and coverage so that more users can access popular files from the cache; when the cache budget is large, it is better to have multiple caches, each with smaller size and coverage, so that users can access files from nearby caches with smaller hit delays.

We also evaluate the algorithms for different values of the service rate of the uncached path assuming the cache budget is fixed at $10,000$.
Figure~\ref{fig:trace_mu} shows the average delay when the ratio of service rate to the total request rate changes from $0.6$ to $1.2$.
Similar to Figure~\ref{fig:trace_cache_size}, the Greedy algorithm significantly reduces the average content access delay. Again, the cache architecture makes little difference for p-LRU but significant affect to the performance of the Greedy algorithm. Moreover, the difference decreases as the service rate on the uncached path increases, as more traffic is offloaded to the uncached path.

\section{Related Work}
\label{sec:relatedWork} 

In this paper, we have considered the {\it joint\/} routing and cache-content management problems.  Numerous past  research efforts have considered these problems separately.
The problem of content placement in caches, has received significant attention in the Internet, in hybrid networks such as those considered in this paper, and in sensor networks~\cite{Tang07, Baev08, Borst10, Sitaraman10, Golrezaei13, Poularakis13}.
Baev~\etal~\cite{Baev08} prove that the problem of content placement with the objective of minimizing the access delay is NP-complete, and present approximate algorithms. 
The separate problem of efficient routing in cache networks
has also been explored in the literature~\cite{Rosensweig_09, Chai_12, Psaras_12}. Rosensweig~\etal~\cite{Rosensweig_09} propose Breadcrumbs -- a simple, best-effort routing policy for locating cached content. Cache-aware routing schemes that calculate paths with  minimum transportation costs based on given caching policy and request demand have been proposed in~\cite{Sourlas13_2}.

The joint caching and routing problem,  with the objective of minimizing content access delay, has recently been studied in~\cite{Golrezaei13, Poularakis13}, where the authors consider a hybrid network consisting of multiple femtocell caches and a cellular infrastructure.
Both papers assume that users greedily choose the minimum delay path to access content, \ie\ requests for cached content are routed to  caches (where content is know to reside), whereas remaining requests are routed to the (uncached) cellular network. They  assume that the delays are constant and independent of the request rate.  

Our work differs from much of the previous research discussed above by considering a joint caching and routing problem, where we determine the optimal routes users should take for accessing content as well as the optimal caching policy.  Our research differs from 
\cite{Golrezaei13, Poularakis13} in
that we consider heterogeneous delays between users and caches, consider a congestion-insensitive delay model for the uncached path as well as a congestion-sensitive model, investigate the problem's time complexity, and propose bounded approximate solutions for both congestion-insensitive and congestion-sensitive scenarios. 
We also determine scenarios for which the optimal solution can be found in polynomial time for the congestion-insensitive delay model, and ascertain the root cause of the complexity of the general problem.

Recent work has also theoretically analyzed the benefits of content caching~\cite{Westphal13, Rosensweig_10, Rosensweig_13, Sourlas13, Sourlas13_2, Gitzenis13}. \cite{Westphal13,Gitzenis13} demonstrate that the asymptotic throughput capacity of a network is significantly increased by adding caching capabilities to the nodes.


\section{Conclusion}
\label{sec:conclusion}

In this paper, we have considered the problem of joint content placement and routing in heterogeneous networks that support in-network caching but also provide a separate, single-hop (uncached) path to a back-end content server; we considered cases in which this uncached path was modeled as a congestion-insensitive, constant-delay path, and  a congestion-sensitive path modeled as an M/M/1 queue. 
We provided fundamental complexity results showing that the problem of joint caching and routing is NP-complete in both  cases, developed a greedy algorithm with guaranteed performance of $(1-1/e)$ of the optimal solution as well as a lower complexity heuristic that was empirically found to provide average delay performance that was within 1\% of optimal (for small instances of the problem) and that significantly reduce the average content access delay over the case of optimized traditional LRU caching.  Our investigation of special-case scenarios $-$ the  congestion-insensitive two-cache case (where we demonstrated an optimal polynomial time solution) and the congestion-sensitive, single-cache, single-file-of-interest case (which we demonstrated remained NP-complete) $-$ helped illuminate what makes the problem ``hard'' in general.  Our future work is aimed at developing a distributed  algorithm for content placement and routing, and on developing solutions for the case of time-varying content popularity.

\bibliographystyle{IEEEtran}
\bibliography{reference}

\begin{appendices}

\section{Network with Two Caches}
\label{sec:appndx_two_cache}

\begin{figure}[t]
\begin{center}
\includegraphics[scale=0.4]{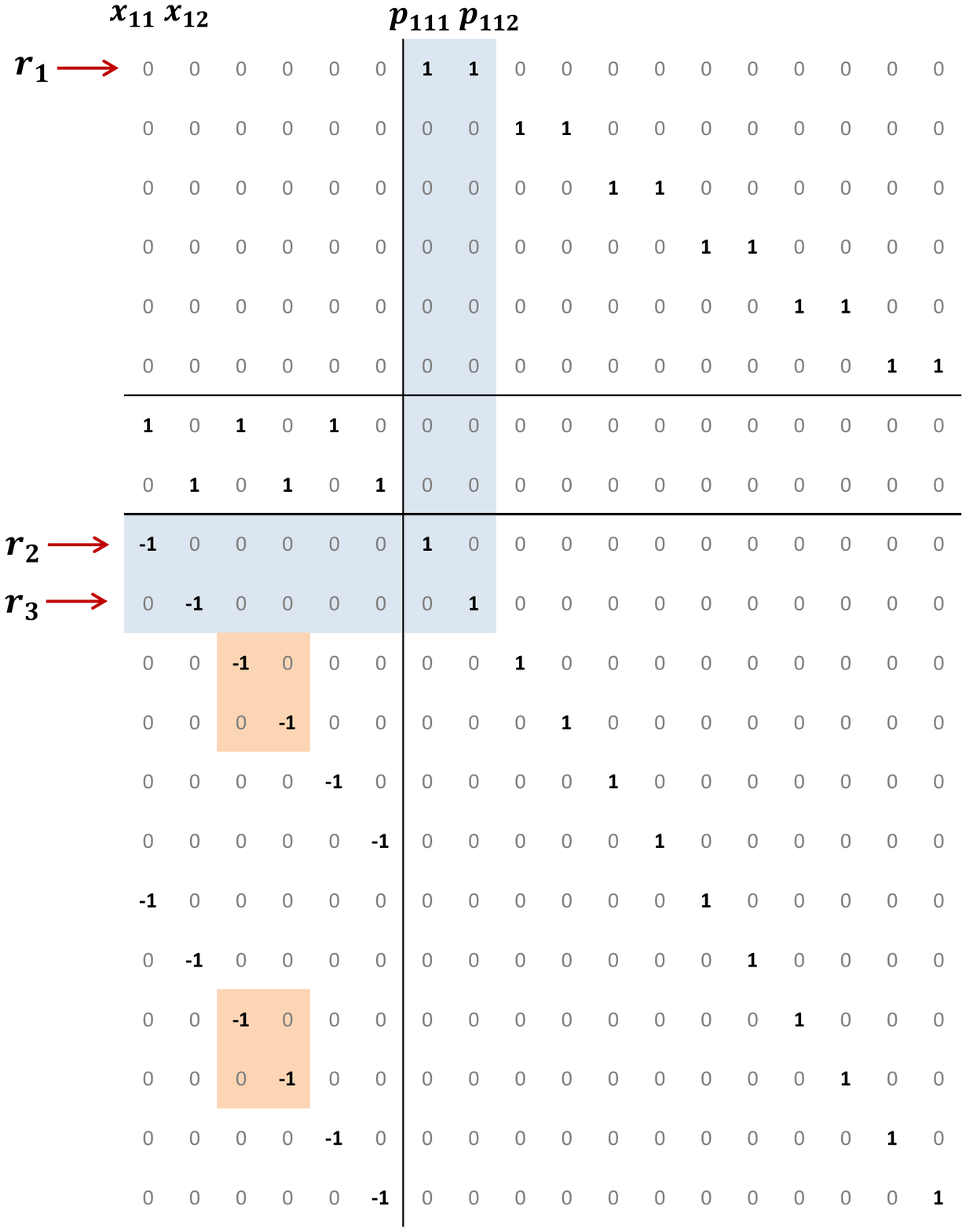}
\caption{An example of the constraints matrix $A$ for a network with two caches, two users and three files}
\label{fig:A2}
\end{center}

\end{figure}
\begin{proof}
Consider the highlighted elements of the matrix in Figure~\ref{fig:A2},
and let $r_1$ denote the first row of the matrix. Also, let $r_2$ and $r_3$ denote
the first two rows below the second horizontal line.
It is easy to see that if these three rows are selected to be in $R$,
any assignment satisfying Proposition~\ref{prop:TU} should have $-s(r_1) = s(r_2) = s(r_3)$.
Otherwise, the signed sum of the rows will have entries other than $\{0, \pm 1\}$.
This observation can be easily extended to see that rows below the second horizontal line
can be considered in groups of two such that if the two rows are selected to be in $R$
they will be assigned the same sign.

We sign the rows in $R$ starting from the rows below the second horizontal line.
Considering the groups of two rows, we make assignments such that the elements
to the left of the vertical line of the signed sum of the rows are in $\{0, -1\}$ only.
To see why this is possible, note that
the non-zero elements of the matrix to the left of the vertical line can be seen
as small blocks of $2\times 2$ matrices. It is easy to see that the signed
sum of any subset of these blocks can be made to have elements only in $\{0, -1\}$,
with rows in the same group getting the same assignment. The rows between the two
horizontal lines are always signed $+1$. The sign of the rows above the first 
horizontal line follows the assignment of the lines below the second horizontal line
based on the previous discussion.

With the above procedure, the sum of the signed vectors will have entries in $\{0,\pm 1\}$
for any set of rows $R$, and from Proposition~\ref{prop:TU} it follows that the matrix $A$ is
totally unimodular, and hence the solution for the optimization problem in~\eqref{eq:opt_simpl_constant}
for a network with two caches can be found in polynomial time.

\end{proof}

\section{Proof of Lemma 2}
\label{sec:mm1_np}
\begin{proof}
It is easy to see that given some $\mathbf{x}, \mathbf{p}$ the expected delay $D(\mathbf{x}, \mathbf{p})$ can be computed in polynomial time, and hence CSDDP is in NP. To show it is NP-hard, we reduce the problem of Equal Cardinality Partition (ECP) to our problem. For an instance of the ECP($A$) problem we create the instance CSDDP($S, A, [\frac{4}{S}], [+\infty], [\frac{4}{a_i}], \frac{n}{2}$) where $S = \sum_{a_i\in A}{a_i}$.

Now, the set $A$ can be partitioned into subsets $A_1$ and $A_2$ with $|A_1|=|A_2|$
if and only if CSDDP achieves delay $(2n+3)/S$.

To see more clearly why the reduction works, first note that with the delay values being set to $d_i^h = \frac{4}{S}$ and $d_i^b = \frac{4}{a_i}$, since $d_i^h < d_i^b$ if a file exits in the cache all the requests for that file will be directed to the cache. Also, since $d_i^m = +\infty$, if a file is not in the cache all the requests for that file will be requested from the back-end server. Therefore, we have $p_i = x_i,\forall{i}$. Now, with the service rate set to $\mu = S$, we can re-write the optimization problem
in~\eqref{eq:single_cache} as follows
\begin{equation}
\label{eq:single_cache_simple}
\begin{aligned}
\text{minimize } & \frac{1}{S}\left[\frac{4}{S}\sum_{i=1}^{n}{a_i x_i} + 4\sum_{i=1}^{n}{(1 - x_i)} +
\frac{S}{\sum_{i=1}^{n}{a_i x_i}} - 1 \right]\\
\text{such that } & \quad\sum_{i=1}^{n}{x_i} \le \frac{n}{2} \\
&\quad x_i \in \{0, 1\}
\end{aligned}
\end{equation}

Now, looking at the objective function in \eqref{eq:single_cache_simple} we can see that
\[Z_1 = 4\sum_{i=1}^{n}{(1 - x_i)} \ge 2n\]
since we should have $\sum_{i=1}^{n}{x_i} \le n/2$. Moreover, $Z_1 = 2n$ if $\sum_{i=1}^{n}{x_i} = n/2$ meaning that exactly half of the files are in the cache. We also have that
\[Z_2 = \frac{4}{S}\sum_{i=1}^{n}{a_i x_i} + \frac{S}{\sum_{i=1}^{n}{a_i x_i}} \ge 4,\]
and $Z_2 = 4$ only if $\sum_{i=1}^{n}{a_i x_i} = S/2$.

Hence, $Z_1 + Z_2 - 1 = 2n+3$ if and only if $\sum_{i=1}^{n}{a_i x_i} = S/2$ and
$\sum_{i=1}^{n}{x_i} = n/2$.

Therefore, if CSDDP($S, A, [\frac{4}{S}], [+\infty], [\frac{4}{a_i}], \frac{n}{2}$)
achieves minimum delay $(2n+3)/S$ then $A$ can be partitioned into equal cardinality subsets.

It is easy to see that if $A$ can be partitioned into two subsets of equal cardinality, then CSDDP($S, A, [\frac{4}{S}], [+\infty], [\frac{4}{a_i}], \frac{n}{2}$) has minimum delay of $(2n+3)/S$.
\end{proof}

\section{Proof of Lemma 4}
\label{sec:appndx_submodularity}
\begin{proof}

Since the congestion-insensitive case can be seen as a special case of the congestion-sensitive model with $\mu = +\infty$, we give a proof for the congestion-sensitive case only.

It is easy to see that placing more content in the cache will not increase the delay, and this implies that $F$ is a monotone increasing function of $\mathbf{x}$.

Writing the Lagrangian for the optimization in~\eqref{eq:optimization} we get:
\[L = -D(\mathbf{p}; \mathbf{x}) + \sum_{i,j}{\nu_{ij}(1-p_{ij})} + \sum_{i,j}{\sigma_{ij}p_{ij}}.\]

Differentiating with respect to $p_{ij}$ we have
\[\lambda_i q_{ij} d_{ij} - \frac{\mu \lambda_i q_{ij}}{(\mu - \sum_{i,j}{\lambda_i q_{ij} p_{ij}})^2} - \nu_{ij} + \sigma_{ij} = 0,\]
or equivalently
\begin{equation}
\label{eq:mm1_diff}
\frac{1}{(\mu - \sum_{i,j}{\lambda_i q_{ij} p_{ij}})^2} = \frac{\lambda_i q_{ij} d_{ij} -\nu_{ij} + \sigma_{ij}}{\mu \lambda_i q_{ij}}.
\end{equation}
Using K.K.T. conditions, the following are necessary:
\begin{equation}
\label{eq:lagrange_derivatives}
\begin{aligned}
&\nu_{ij} (1-p_{ij}) = 0 \\
&\sigma_{ij}p_{ij} = 0 \\
&p_{ij} \le 1 \\
&-p_{ij} \le 0 \\
&\nu_{ij}, \sigma_{ij} \ge 0
\end{aligned}
\end{equation}
From the above conditions, it is clear that if $p_{ij} = 1$, then $\sigma_{ij} = 0$.
Also, if $p_{ij} = 0$, then $\nu_{ij} = 0$, and if $0 < p_{ij} < 1$ then $\sigma_{ij} = 0$ and $\nu_{ij} = 0$.
Therefore,~\eqref{eq:mm1_diff} simplifies to
\[\frac{1}{(\mu - \sum_{i,j}{\lambda_i q_{ij} p_{ij}})^2} =
 \left\lbrace \begin{array}{lc}
 \frac{d_{ij}}{\mu} - \frac{\nu_{ij}}{\mu\lambda_i q_{ij}} & p_{ij} = 1 \\
 \frac{d_{ij}}{\mu} & 0 < p_{ij} < 1 \\ 
 \frac{d_{ij}}{\mu} + \frac{\sigma_{ij}}{\mu\lambda_i q_{ij}} & p_{ij} = 0
 \end{array} \right.\]

Since $\nu_{ij}, \sigma_{ij} \ge 0$, there should exist some $d^*$ such that
if $d_{ij} < d^*$ then $p_{ij} = 0$, and if $d_{ij} > d^*$ then $p_{ij} = 1$.
Moreover,
\[\frac{1}{(\mu - \sum_{i,j}{\lambda_i q_{ij} p_{ij}})^2} = \frac{d^*}{\mu},\]
or equivalently
\begin{equation}
\label{eq:mm1_load_d}
\mu - \sum_{i,j}{\lambda_i q_{ij} p_{ij}} = \sqrt{\frac{\mu}{d^*}}.
\end{equation}

Let $X = \{x_{jm}\}$ denote the set of files in the caches. Since adding an item to the
cache will not increase $d_{ij}$, we have the following lemma:
\begin{lemma}
\label{lem:d_decreasing}
For two content placement sets $X$ and $Y$ if $X \subseteq Y$ then $d^*_X \ge d^*_Y$.
\end{lemma}

Note that based on the solution structure for the traffic forwarded to the uncached path we have that
\[\mu - \sum_{i,j}{\lambda_i q_{ij} p_{ij}} = \mu - \sum_{i,j}{\lambda_i q_{ij}  \mathbbm{1}\{d_{ij}  > d^*\}} - p^* \sum_{i,j}{\lambda_i q_{ij}},\]
or equivalently
\begin{equation}
\label{eq:traffic_split}
p^* \sum_{i,j}{\lambda_i q_{ij} \mathbbm{1}\{d_{ij} = d^*\}} = \mu - \sqrt{\frac{\mu}{d^*}} - \sum_{i,j}{\lambda_i q_{ij} \mathbbm{1}\{d_{ij} > d^*\}}
\end{equation}

Using~\eqref{eq:mm1_load_d} and~\eqref{eq:traffic_split} we can simplify the delay function as
\begin{align*}
D &= \sum_{i,j}{\lambda_i q_{ij} d_{ij} \mathbbm{1}\{d_{ij} < d^*\}} \\
&\quad + (1-p^*) d^* \sum_{i,j}{\lambda_i q_{ij} \mathbbm{1}\{d_{ij} = d^*\}} + \mu\sqrt{\frac{d^*}{\mu}} - 1 \\
&= \sum_{i,j}{\lambda_i q_{ij} d_{ij} \mathbbm{1}\{d_{ij} < d^*\}} + d^* \sum_{i,j}{\lambda_i q_{ij} \mathbbm{1}\{d_{ij} = d^*\}} \\
 &\quad - d^* (\mu - \sqrt{\frac{\mu}{d^*}} - \sum_{i,j}{\lambda_i q_{ij} \mathbbm{1}\{d_{ij} > d^*\}}) + \sqrt{\mu d^*} - 1 \\
&= -\mu d^* + 2\sqrt{\mu d^*} - 1 + \sum_{i,j}{\lambda_i q_{ij} d_{ij}\mathbbm{1}\{d_{ij} \le d^*\}} \\
&\quad + \sum_{i,j}{\lambda_i q_{ij} d^*\mathbbm{1}\{d_{ij} > d^*\}} \\
&= -(\sqrt{\mu d^*} - 1)^2 + \sum_{i,j}{\lambda_i q_{ij} \min\{d^*, d_{ij}\}}.
\end{align*}

Let $X=\{x_{jm}\}$ denote the set of files in the caches and define
\begin{equation}
\label{eq:fdX}
f(d; X)= (\sqrt{\mu d} - 1)^2 - \sum_{i,j}{\lambda_i q_{ij} \min\{d, d_{ij}\}}.
\end{equation}
Let $d^X_{ij}$ and $d^Y_{ij}$ denote the cache access delay for user $i$ for file $j$
given content placement $X$ and $Y$, respectively.
If $X\subseteq Y$, then $d^X_{ij} \ge d^Y_{ij}$ and hence
\begin{equation}
\label{eq:f_diff}
f(d; Y) - f(d; X) = \sum_{i,j}{\lambda_i q_{ij} \left(\min\{d, d^X_{ij}\} - \min\{d, d^Y_{ij}\}\right)}
\end{equation}
is a non-decreasing function of $d$. The following lemma summarizes this result.

\begin{lemma}
\label{lem:inc_d_diff}
For two sets $X \subseteq Y$ and for $d_1 \le d_2$ we have
\[f(d_2; Y) - f(d_2; X) \ge f(d_1; Y) - f(d_1; X).\]
\end{lemma}

Next, we consider the function
\[-D_X = f(d^*_X; X) \triangleq \min_{d}{f(d; X)}.\]
Consider two sets $X$ and $Y$ such that $X \subseteq Y$, then $d^*_X \ge d^*_Y$ due
to Lemma~\ref{lem:d_decreasing}.
Now let $\Delta(X; x_{km}) = f(d^*_{X \cup x_{km}}; X\cup x_{km}) - f(d^*_X; X)$
and $\Delta(Y; x_{km}) = f(d^*_{Y \cup x_{km}}; Y\cup x_{km}) - f(d^*_Y; Y)$
denote the gain obtained by adding $x_{km}$ to the sets $X$ and $Y$, respectively.

From Lemma~\ref{lem:d_decreasing} we have
\[\Delta(X; x_{km}) \ge f(d^*_{Y \cup x_{km}}; X\cup x_{km}) - f(d^*_X; X), \]
and
\[\Delta(Y; x_{km}) \le f(d^*_{Y \cup x_{km}}; Y\cup x_{km}) - f(d^*_X; Y). \]
Therefore, we have
\begin{align*}
&\Delta(X; x_{km}) - \Delta(Y; x_{km}) \\
& \quad \ge \left[ f(d^*_{Y \cup x_{km}}; X\cup x_{km}) - f(d^*_{Y \cup x_{km}}; Y\cup x_{km}) \right] \\
&\quad \quad \quad \quad \quad - \left[ f(d^*_X; X) - f(d^*_X; Y) \right].
\end{align*}

Moreover, from Lemma~\ref{lem:d_decreasing} we have
\begin{align*}
&\Delta(X; x_{km}) - \Delta(Y; x_{km}) \\
&\quad \ge \left[ f(d^*_X; X\cup x_{km}) - f(d^*_X; Y\cup x_{km}) \right] \\
& \quad \quad \quad \quad - \left[ f(d^*_X; X) - f(d^*_X; Y) \right] \\
&\quad = \left[ f(d^*_X; X\cup x_{km}) - f(d^*_X; X) \right] \\
&\quad \quad - \left[ f(d^*_X; Y\cup x_{km}) - f(d^*_X; Y) \right] \\
& \quad = \sum_{i,j}{\lambda_i q_{ij} \left(\min\{d^*_X, d^X_{ij}\} - \min\{d^*_X, d^{X\cup x_{km}}_{ij}\} \right)} \\
&\qquad - \sum_{i,j}{\lambda_i q_{ij} \left(\min\{d^*_X, d^Y_{ij}\} - \min\{d^*_X, d^{Y\cup x_{km}}_{ij}\} \right)} \\
&\quad \ge 0,
\end{align*}
where the last inequality follows from the fact that $d^X_{ij} - d^{X\cup x_{km}}_{ij} \ge d^Y_{ij} - d^{Y\cup x_{km}}_{ij}$ if $X \subseteq Y$. Note that $d^{X\cup x_{km}}_{ij} = \min\{d^X_{ij}, d^h_{im}\}$.

From $\Delta(X; x_{km}) - \Delta(Y; x_{km}) \ge 0$ we conclude that $F(x)$ is submodular.
\end{proof}

\end{appendices}

\end{document}